\documentclass[12pt]{article}
\usepackage[T1]{fontenc}
\usepackage[latin9]{inputenc}
\usepackage{float}
\usepackage{amsmath}
\usepackage{amssymb}
\usepackage{graphicx}
\usepackage{esint}

\makeatletter
\@ifundefined{date}{}{\date{}}

\usepackage{epsfig}\usepackage{euscript}\usepackage{latexsym}\usepackage{amsthm}\usepackage{color}\usepackage{multirow}\usepackage{epsfig}

\theoremstyle{remark}

\theoremstyle{plain}
\newtheorem{theorem}{Theorem}

\newcommand{\fig}[1]{Fig. \ref{#1}}

\newcommand{\theor}[1]{Theorem \ref{#1}}

\makeatother

\begin{document}

\title{\vspace{1cm}
 \textbf{On stochastic generation of ultrametrics in high-dimensional
Euclidean spaces} }

\author{A.\,P.~Zubarev\bigskip{}
 \\
 \textit{Physics Department, Samara State University of Railway Transport,
} \\
 \textit{Perviy Bezimyaniy pereulok 18, 443066, Samara, Russia} \medskip{}
 \\
 \textit{Physics Department, Samara State Aerospace University, } \\
 \textit{Moskovskoe shosse 34, 443123, Samara, Russia} \medskip{}
 \\
 e-mail:\:\texttt{apzubarev@mail.ru}}
\maketitle
\begin{abstract}
We present a proof of the theorem which states that a matrix of Euclidean
distances on a set of specially distributed random points in the $n$-dimensional
Euclidean space $R^{n}$ converges in probability to an ultrametric
matrix as $n\rightarrow\infty$. Values of the elements of an ultrametric
distance matrix are completely determined by variances of coordinates
of random points. Also we preset a probabilistic algorithm for generation
of finite ultrametric structures of any topology in high-dimensional
Euclidean space. Validity of the algorithm is demonstrated by explicit
calculations of distance matrices and ultrametricity indexes for various
dimensions $n$.

\end{abstract}

\section{Introduction}

A set $M=\left\{ x\right\} $ is called a metric space if for any
pair of elements $x^{a},\, x^{b}\in M$ a distance function $d_{ab}=d(x^{a},x^{b})$
(called a metric) is defined and satisfies the following four axioms
for any triplet of points $x^{a},\, x^{b},\, x^{c}$:
\begin{equation}
d_{ab}\ge0,\label{m_1}
\end{equation}
\begin{equation}
d_{ab}=0\,\Leftrightarrow\, a=b,\label{m_2}
\end{equation}
\begin{equation}
d_{ab}=d_{ab},\label{m_3}
\end{equation}
\begin{equation}
d_{ab}\le d_{ac}+d_{dc}.\label{m_4}
\end{equation}
A metric $d_{ab}$ satisfying the strong triangle inequality
\begin{equation}
d_{ab}\le\max\{d_{ac},\: d_{bc}\}\label{um}
\end{equation}
is called an ultrametric. A space with an ultrametric is called a
space with an ultrametric structure or an ultrametric space.

An arbitrary real valued matrix $d=\{d_{ab}\}$ is called \emph{a
metric matrix} if its elements satisfy the conditions (\ref{m_1})--(\ref{m_4}),
and it is called \emph{an ultrametric matrix} if its elements satisfy
the conditions (\ref{m_1})--(\ref{m_3}) and (\ref{um}).

Classic examples of ultrametric spaces are the field $Q_{p}$ of $p$-adic
numbers \cite{S,VVZ,Sh_Kh}, and the ring $Q_{m}$ of $m$-adic numbers
(see, for example, \cite{DZ}).

Ultrametric spaces had long been used in various fields of natural
and social sciences to problems of classification and information
processing: optimization theory, taxonomy, cluster and factor analysis,
and other \cite{RTV}. During the last 30 years, the mathematical
apparatus of ultrametric analysis was developed by by V.S.Vladimirov
and coworkers. This relatively new scientific field of ultrametric
mathematical physics is represented by many books and works devoted
to development of the $p$-adic analysis, $p$-adic mathematical physics
and its application to modeling in various fields of physics, biology,
computer science, psychology, sociology, and so on (see \cite{S,VVZ,Sh_Kh,ALL}
and references therein).

In fact a lot of physical, biological or socio-economic systems has
intrinsic hierarchical structure \cite{RTV}. Systems with non--explicit
hierarchical structure are of considerable interest to researchers.
In such systems the hierarchical structure cannot be observed in the
original variables, but it becomes observable after transition to
some effective (hidden) variables. Typically the number of these effective
variables is essentially smaller than the number of degree of freedom
of the whole system. There are reasons to believe that hidden ultrametric
structures are present in a number of complex systems (i.e. systems
with a large number of heterogeneous interacting objects), which include
spin glasses, proteins, nucleic acids, etc. Similar ultrametric models
arose in the early eighties last century in statistical physics of
spin systems with a disorder \cite{MPSTV,RTV,Parisi1,Dayson,Dot}.
Namely it was found that if the system has large number of \textquotedbl{}internal
contradictions\textquotedbl{} (frustrations) at different scales,
equilibrium of the system can be achieved in hierarchically nested
regions of phase space, and number of nesting levels increases with
decreasing a temperature. In this case there are relations between
phase region scales, that satisfy the strong triangle inequality and
thus low-temperature spin states are correlated ultrametrically.

Almost immediately after the appearance of ultrametric spin glass
models it has been assumed that the conformational state space of
protein molecule has ultrametric structure \cite{Frauen1,Frauen2}.
In this case states are associated with local minima of the potential
energy landscape of a protein molecule, and the energy landscape is
represented as a hierarchy of nested basins of free energy local minima.
Ultrametric models of conformational dynamics of protein molecules
have been developed in a series of papers \cite{OS,BK,ABK,ABKO,ABO,AB,ABZ}.

There is also some evidence that similar ultrametric structures arise
in socio-economic systems \cite{SJ,MS,BZK,V,VZ}.

In many cases the idea of application of ultrametric models to complex
systems such as proteins came from the theory of spin glasses \cite{Frauen1,Frauen2}
and there are many arguments in favor of that protein should have
the ultrametric structure (see \cite{BK} for the discussion of energy
landscapes and hierarchical disconnectivity graphs). However, the
explanation of the origin of ultrametric structures in spin glasses
based on the replica method, which is not quite rigorously justified.
Thus it is interesting to discuss alternative approaches to ultrametric
structures in complex systems.

In this paper we propose a procedure of generation of ultrametric
structure in a metric space. In this procedure we do not assume any
ultrametric properties for the initial metric space.

It has been observed in several studies (see, for example, \cite{Hall,Murtagh}),
that the effectiveness of clustering algorithms applied to large data
sets significantly increases with increasing the dimension of the
array. Moreover it has been observed that the distance between randomly
distributed points in multidimensional metric spaces shows ultrametric
properties with increasing space dimension. In this paper we give
the rigorous proof that any finite ultrametric space can be generated
by a special random distribution of points in the $n$-dimensional
Euclidean metric space taking the limit $n\rightarrow\infty$. In
this case the ultrametric is completely determined by variances of
random point coordinates. We present the algorithm for generating
such ultrametric spaces. The validity of the algorithm is numerically
demonstrated by calculations of distance matrices and ultrametricity
indexes for high dimension spaces.

The paper is organized as follows. In Section 2, we present the formulation
of our construction. Also in this section we formulate and prove the
theorem that is the main result of this article. In Section 3, we
describe the algorithm of a stochastic generation of ultrametrics
in high-dimensional Euclidean spaces and check the validity of its
by numerical simulations.

\section{Stochastic generation of ultrametric matrices}

Let us formulate statements which will be used in our construction.

Let $M=\left\{ x^{(a)}\right\} ,\: a=1,2,\ldots,N$ be a finite ultrametric
space with an ultrametric $d\left(x^{(a)},x^{(b)}\right)$. We say
that a space $M$ is \emph{an isometric space}, and an ultrametric
$d\left(x^{(a)},x^{(b)}\right)$ is \emph{an isometric} if for any
triplet of points $x^{(a)},\, x^{(b)},\, x^{(c)}$ in $M$ the following
condition holds: $d\left(x^{(a)},x^{(b)}\right)=\max\left\{ d\left(x^{(a)},x^{(c)}\right),d\left(x^{(b)},x^{(c)}\right)\right\} $,
i.e. distances between any two non--coinciding points are equal.

The subset $B_{r}(a)=$ $\left\{ x\in M:\: d\left(x,a\right)\leq r\right\} $
is said to be the ball of radius $r$ with the center at the point
$a$.


We say that an ultrametric space $M$ is \emph{homogeneous}, if for
any fixed value $r$ of ball radius there is exist the number $m(r)$
such that any ball $B_{r}(a)$ can be represented as a union of $m(r)$
balls of radius $r'$, $r'<r$.

We say that an ultrametric space $M$ is \emph{self-similar}, if there
exist is the number $m$ such that any ball $B_{r}(a)$ can be represented
as a union of $m$ balls of radius $r'$, $r'<r$.

It is obvious that any self-similar ultrametric space is homogeneous.

A finite self-similar ultrametric space is isomorphic to a boundary
of a Cayley tree with a finite number of levels. The distance between
points at the boundary of the tree is defined as the weighted length
of the path in the tree between these points.

Also we need some statements from probability theory (see, for example,
\cite{Dur,Bill1,Shiryaev}). These results will be used to prove the
main result formulated in Theorem 7.

Let $\left\{ \Omega,\Sigma,\mathrm{P}\right\} $ be a probability
space, where \textbf{$\left\{ \Omega,\Sigma\right\} $ } is measurable
space, $\mathrm{P}$ is probability measure. Real random variable
$X$ is measurable mapping $X:\:\Omega\to R$. For any real random
variable $X=X(\omega)$ an interegral $\intop_{A}X\left(\omega\right)d\mathrm{P}\left(\omega\right)$,
$A\in\Sigma$ can be defined. An expectation and a variance of $X$
are $\mathsf{E}\left[X\right]=\intop_{\Omega}X\left(\omega\right)d\mathrm{P}\left(\omega\right)$
and $\mathsf{V}\left[X\right]=\mathsf{E}\left[X^{2}\right]-\left(\mathsf{E}\left[X\right]\right)^{2}$
respectively. Let $\Sigma^{(1)}\subset\Sigma$ be a $\sigma$-subalgebra
of $\Sigma$, then the conditional expectation $\mathsf{E}\left[X\left|\Sigma^{(1)}\right.\right]$
of real random variable $X$ is a random variable $Y$ that that $Y$
is $\Sigma^{(1)}$ measurable, and for all $A\in\Sigma^{(1)}$ $\intop_{A}X\left(\omega\right)d\mathrm{P}\left(\omega\right)=\intop_{A}Y\left(\omega\right)d\mathrm{P}\left(\omega\right)$,
and the conditional variance $\mathsf{V}\left[X\left|\Sigma^{(1)}\right.\right]$
is defined as $\mathsf{V}\left[X\left|\Sigma^{(1)}\right.\right]=\mathsf{E}\left[X^{2}\left|\Sigma^{(1)}\right.\right]-\left(\mathsf{E}\left[X\left|\Sigma^{(1)}\right.\right]\right)^{2}$.

\begin{theorem} (The law of large numbers) Let $X_{1},X_{2},\ldots$
be a sequence of independent identically distributed random variables
with finite expectations $\mathsf{E}\left[X_{i}\right]\equiv m_{i}$
and finite variances $\mathsf{V}\left[X_{i}\right]$, the variances
$\mathsf{V}\left[X_{i}\right]$ are uniformly bounded, and $S_{n}=X_{1}+X_{2}+\ldots+X_{n}$.
Then $\dfrac{S_{n}}{n}\overset{\mathrm{P}}{\rightarrow}\dfrac{\left\langle S_{n}\right\rangle }{n}$
i.e. for any $\varepsilon>0$ one has $\mathrm{P}\left\{ \left|\dfrac{S_{n}}{n}-\dfrac{\left\langle S_{n}\right\rangle }{n}\right|\geq\varepsilon\right\} \rightarrow0$
as $n\rightarrow\infty$ (convergence in probability). \label{th1}
\end{theorem}

\begin{theorem} (Slutsky's theorem, \cite{Sl,Bill}) If $X_{n}^{(1)}\overset{\mathrm{P}}{\rightarrow}X^{(1)}$,
$X_{n}^{(2)}\overset{\mathrm{P}}{\rightarrow}X^{(2)}$, $\ldots,$
$X_{n}^{(N)}\overset{\mathrm{P}}{\rightarrow}X^{(N)}$, and $h\left(x^{(1)},x^{(2)},\ldots,x^{(N)}\right)$
is a continuous function of $N$ variables, then $h\left(X_{n}^{(1)},X_{n}^{(2)},\ldots,X_{n}^{(N)}\right)\overset{\mathrm{P}}{\rightarrow}h\left(X^{(1)},X^{(2)},\ldots,X^{(N)}\right)$.
\label{th2} \end{theorem}

\begin{theorem} Let $M=\left\{ x\right\} $ be an ultrametric space
with an ultrametric $d(x,y)$, and let $f\left(\lambda\right)$ be
a continuous nonnegative nondecreasing function such that $f(0)=0$.
Then $\delta(x,y)\equiv f\left(d(x,y)\right)$ is an ultrametric on
$M$. \label{th3} \end{theorem}

\begin{theorem} Let $\left\{ \Omega,\Sigma,\mathrm{P}\right\} $
be probability space, and $x^{(a)}=\left(x_{1}^{(a)},x_{2}^{(a)},\ldots,x_{n}^{(a)}\right)$,
$x^{(b)}=\left(x_{1}^{(b)},x_{2}^{(b)},\ldots,x_{n}^{(b)}\right)$
be two points in the $n$-dimencional space $R^{n}$ with independent
random coordinates on $\left\{ \Omega,\Sigma,\mathrm{P}\right\} $
. Suppose coordinates of point $x^{(a)}$ have the finite expectation
$\mathsf{E}\left[x_{i}^{(a)}\right]=m_{a}$ and the finite variance
$\mathsf{V}\left[x_{i}^{(a)}\right]=\sigma_{a}^{2}$, and coordinates
of the point $x^{(b)}$ have the finite expectation $\mathsf{E}\left[x_{i}^{(b)}\right]=m_{b}$
and the finite variance $\mathsf{V}\left[x_{i}^{(b)}\right]=\sigma_{b}^{2}$.
Then the distance between these points
\begin{equation}
d_{n}\left(x^{(a)},x^{(b)}\right)\equiv\dfrac{1}{\sqrt{n}}\sqrt{\sum_{i=1}^{n}\left(x_{i}^{(a)}-x_{i}^{(b)}\right)^{2}}\label{d}
\end{equation}
satisfies the condition
\begin{equation}
d_{n}\left(x^{(a)},x^{(b)}\right)\overset{\mathrm{P}}{\rightarrow}\sqrt{\sigma_{a}^{2}+\sigma_{b}^{2}+\left(m_{a}-m_{b}\right)^{2}}\label{eq_theorem_4}
\end{equation}
as $n\to\infty$. \label{th4} \end{theorem}

\begin{proof} Consider the random variables $\left(x_{i}^{(a)}-x_{i}^{(b)}\right)^{2}$
(for all $i=1,2,\ldots,n$). Their expectations are
\[
\mathsf{E}\left[\left(x_{i}^{(a)}-x_{i}^{(b)}\right)^{2}\right]=\mathsf{E}\left[x_{i}^{(a)2}\right]+\mathsf{E}\left[x_{i}^{(a)2}\right]-2\mathsf{E}\left[x_{i}^{(a)}\right]\mathsf{E}\left[x_{i}^{(a)2}\right]=
\]
\[
=\left(\sigma_{a}^{2}+m_{a}^{2}\right)+\left(\sigma_{b}^{2}+m_{b}^{2}\right)-2m_{a}m_{b}=\sigma_{a}^{2}+\sigma_{b}^{2}+\left(m_{a}-m_{b}\right)^{2}.
\]
By the law of large numbers we obtain that
\[
\dfrac{\sum_{i=1}^{n}\left(x_{i}^{(a)}-x_{i}^{(b)}\right)^{2}}{n}\overset{\mathrm{P}}{\rightarrow}\sigma_{a}^{2}+\sigma_{b}^{2}+\left(m_{a}-m_{b}\right)^{2}
\]
as $n\rightarrow\infty$. Using Slutsky's theorem, we get (\ref{eq_theorem_4}).
\end{proof}

The following theorem is a direct consequence of \theor{th4}.

\begin{theorem} Let \textbf{$\left\{ \Omega,\Sigma,\mathrm{P}\right\} $
}be probability space, and $M_{n}\equiv\left\{ x^{(a)}\right\} $
($a=1,2,\ldots,N$ ) be the set of points $x^{a}=\left(x_{1}^{a},x_{2}^{a},\ldots,x_{n}^{a}\right)$
in $R^{n}$ with independent random coordinates on $\left\{ \Omega,\Sigma,\mathrm{P}\right\} $.
Suppose coordinates of point $x^{(a)}$ have the finite expectation
$\mathsf{E}\left[x_{i}^{(a)}\right]=m_{i}$ and the finite variances
$\mathsf{V}\left[x_{i}^{(a)}\right]=\sigma^{2}$. Then the metric
on $M_{n}$
\[
d_{n}\left(x^{(a)},x^{(b)}\right)\equiv\dfrac{1}{\sqrt{n}}\sqrt{\sum_{i=1}^{n}\left(x_{i}^{(a)}-x_{i}^{(b)}\right)^{2}}
\]
satisfies the condition
\[
d_{n}\left(x^{(a)},x^{(b)}\right)\overset{\mathrm{P}}{\longrightarrow}u_{ab},
\]
as $n\to\infty$, where
\[
u_{ab}=\left\{ \begin{array}{c}
\sqrt{2}\sigma,\: a\neq b,\\
0,\: a=b.
\end{array}\right.
\]
\label{th5} \end{theorem}

We also have the following theorem.

\begin{theorem} Let \textbf{$\left\{ \Omega,\Sigma,\mathrm{P}\right\} $}
be probability space, $\Sigma^{(1)}\subset\Sigma$ be a $\sigma$-subalgebra
of $\Sigma$. Let $M_{n}\equiv\left\{ x^{(a_{1}a_{2})}\right\} $
(here $a_{1}=1,2,\ldots,p_{1},\: a_{2}=1,2,\ldots,p_{2}$, and $a_{1}a_{2}$
is two-dimensional index) be the set of $p_{1}p_{2}$ points $x^{(a_{1}a_{2})}=\left(x_{1}^{(a_{1}a_{2})},x_{2}^{(a_{1}a_{2})},\ldots,x_{n}^{(a_{1}a_{2})}\right)$
in $R^{n}$ with independent random coordinates on $\left\{ \Omega,\Sigma,\mathrm{P}\right\} $.
Suppose conditional expectations $\mathsf{E}\left[x_{i}^{(a_{1}a_{2})}\left|\Sigma^{(1)}\right.\right]\equiv x_{i}^{(a_{1})}$
are identical for all $a_{2}$, and $\mathsf{E}\left[x_{i}^{\left(a_{1}a_{2}\right)}\right]=m_{i}$,
$\mathsf{V}\left[x_{i}^{\left(a_{1}\right)}\right]=\sigma_{1}^{2}$,
$\mathsf{E}\left[\mathsf{V}\left[\left(x_{i}^{(a_{1}a_{2})}\right)\left|\Sigma^{(1)}\right.\right]\right]=\sigma_{2}^{2}$
with real and finite $m_{i}$, $\sigma_{1}^{2}$, $\sigma_{2}^{2}$.
Then the metric on $M_{n}$
\begin{equation}
d_{n}\left(x^{(a_{1}a_{2})},x^{(b_{1}b_{2})}\right)\equiv\dfrac{1}{\sqrt{n}}\sqrt{\sum_{i=1}^{n}\left(x_{i}^{(a_{1}a_{2})}-x_{i}^{(b_{1}b_{2})}\right)^{2}}\label{d_abcd}
\end{equation}
satisfies the condition
\[
d_{n}\left(x^{(a_{1}a_{2})},x^{(b_{1}b_{2})}\right)\overset{\mathrm{P}}{\longrightarrow}u_{a_{1}a_{2},b_{1}b_{2}}
\]
as $n\to\infty$, where
\begin{equation}
u_{a_{1}a_{2},b_{1}b_{2}}=\left\{ \begin{array}{c}
\sqrt{2}\left(\sigma_{1}^{2}+\sigma_{2}^{2}\right)^{\tfrac{1}{2}},\: a_{1}\neq b_{1},\: a_{2}\neq b_{2},\\
\sqrt{2}\sigma_{1},\: a_{1}=b_{1},\: a_{2}\neq b_{2},\\
0,\: a_{1}=b,\: a_{2}=b_{2}
\end{array}\right.\label{u_aabb}
\end{equation}
is nonisometric ultrametric $p_{1}p_{2}\times p_{1}p_{2}$ matrix.
\label{th6} \end{theorem}

\begin{proof}

Consider the random variables $\left(x_{i}^{(a_{1}a_{2})}-x_{i}^{(b_{1}b_{2})}\right)^{2}$.
Their expectations are
\[
\mathsf{E}\left[\left(x_{i}^{(a_{1}a_{2})}-x_{i}^{(b_{1}b_{2})}\right)^{2}\right]=\left(1-\delta_{a_{1}b_{1}}\delta_{a_{2}b_{2}}\right)\mathsf{E}\left[\mathsf{E}\left[\left(x_{i}^{(a_{1}a_{2})}-x_{i}^{(b_{1}b_{2})}\right)^{2}\left|\Sigma^{(1)}\right.\right]\right]=
\]
\[
=\left(1-\delta_{a_{1}b_{1}}\delta_{a_{2}b_{2}}\right)\mathsf{E}\left[\mathsf{E}\left[\left(x_{i}^{(a_{1}a_{2})}\right)^{2}\left|\Sigma^{(1)}\right.\right]+\right.
\]

\[
\left.\mathsf{+E}\left[\left(x_{i}^{(b_{1}b_{2})}\right)^{2}\left|\Sigma^{(1)}\right.\right]-2\mathsf{E}\left[x_{i}^{(a_{1}a_{2})}\left|\Sigma^{(1)}\right.\right]\mathsf{E}\left[x_{i}^{(b_{1}b_{2})}\left|\Sigma^{(1)}\right.\right]\right]=
\]

\[
=\left(1-\delta_{a_{1}b_{1}}\delta_{a_{2}b_{2}}\right)\mathsf{E}\left[\mathsf{V}\left[\left(x_{i}^{(a_{1}a_{2})}\right)\left|\Sigma^{(1)}\right.\right]+\mathsf{V}\left[\left(x_{i}^{(b_{1}b_{2})}\right)\left|\Sigma^{(1)}\right.\right]+\right.
\]
\[
\left.+\left(1-\delta_{a_{1}b_{1}}\delta_{a_{2}b_{2}}\right)\left(1-\delta_{a_{1}b_{1}}\right)\left(\mathsf{E}\left[\left(x_{i}^{(a_{1}a_{2})}\right)\left|\Sigma^{(1)}\right.\right]-\mathsf{E}\left[\left(x_{i}^{b{}_{1}b_{2})}\right)\left|\Sigma^{(1)}\right.\right]\right)^{2}\right]=
\]

\[
=\left(1-\delta_{a_{1}b_{1}}\delta_{a_{2}b_{2}}\right)\left(\mathsf{E}\left[\mathsf{V}\left[\left(x_{i}^{(a_{1}a_{2})}\right)\left|\Sigma^{(1)}\right.\right]\right]+\mathsf{E}\left[\mathsf{V}\left[\left(x_{i}^{(b_{1}b_{2})}\right)\left|\Sigma^{(1)}\right.\right]\right]\right)+
\]
\[
+\left(1-\delta_{a_{1}b_{1}}\right)\left(\mathsf{V}\left[\mathsf{E}\left[x_{i}^{\left(a_{1}a_{2}\right)}\right]\left|\Sigma^{(1)}\right.\right]+\mathsf{V}\left[\mathsf{E}\left[x_{i}^{\left(b_{1}b_{2}\right)}\right]\left|\Sigma^{(1)}\right.\right]\right)+
\]

\[
+\left(1-\delta_{a_{1}b_{1}}\right)\left(\mathsf{E}\left[\left(x_{i}^{(a_{1}a_{2})}\right)\right]-\mathsf{E}\left[\left(x_{i}^{b{}_{1}b_{2})}\right)\right]\right)^{2}=
\]

\[
=2\left(\left(1-\delta_{a_{1}b_{1}}\delta_{a_{2}b_{2}}\right)\sigma_{2}^{2}+\left(1-\delta_{a_{1}b_{1}}\right)\sigma_{1}^{2}\right),
\]
where $\delta_{ab}=\left\{ \begin{array}{c}
1,\: a=b,\\
0,\: a\neq b
\end{array}\right.$ is Kronecker delta. By the law of large numbers, we obtain

\[
\dfrac{\sum_{i=1}^{n}\left(x_{i}^{(a_{1}a_{2})}-x_{i}^{(b_{1}b_{2})}\right)^{2}}{n}\overset{\mathrm{P}}{\longrightarrow}2\left(\left(1-\delta_{a_{1}b_{1}}\delta_{a_{2}b_{2}}\right)\sigma_{2}^{2}+\left(1-\delta_{a_{1}b_{1}}\right)\sigma_{1}^{2}\right).
\]
as $n\to\infty$. By Slutsky's theorem, we get
\[
d_{n}\left(x^{(a_{1}a_{2})},x^{(b_{1}b_{2})}\right)\overset{\mathrm{P}}{\longrightarrow}u_{a_{1}a_{2},b_{1}b_{2}}=\sqrt{2}\left(\left(1-\delta_{a_{1}b_{1}}\delta_{a_{2}b_{2}}\right)\sigma_{2}^{2}+\left(1-\delta_{a_{1}b_{1}}\right)\sigma_{1}^{2}\right)^{\tfrac{1}{2}}.
\]
\end{proof}

We state the following generalization of \theor{th6} which can be
proved similarly.

\begin{theorem} Let \textbf{$\left\{ \Omega,\Sigma,\mathrm{P}\right\} $}
be probability space, $\Sigma^{(n)}$ be an increasing sequence of
$\sigma$-subalgebras $\Sigma^{(1)}\subset\Sigma^{(2)}\subset\cdots\subset\Sigma^{(N)}\equiv\Sigma$.
Let $M_{n}\equiv\left\{ x^{(a_{1}a_{2}\cdots a_{N})}\right\} $ ($a_{1}=1,2,\ldots,p_{1}$,
$a_{2}=1,2,\ldots,p_{2}$, $\ldots,$ $a_{N}=1,2,\ldots,p_{N}$, and
$a_{1}a_{2}\cdots a_{N}$ is $N$-dimensional index) be the sets of
$p_{1}$, $p_{1}p_{2}$, $\ldots,$ $p_{1}p_{2}\cdots p_{N}$ points
$x^{(a_{1}a_{2}\cdots a_{N})}=\left(x_{1}^{(a_{1}a_{2}\cdots a_{N})},x_{2}^{(a_{1}a_{2}\cdots a_{N})},\ldots,x_{n}^{(a_{1}a_{2}\cdots a_{N})}\right)$
in $R^{n}$ with independent random coordinates. Suppose conditional
expectations $\mathsf{E}\left[x_{i}^{(a_{1}a_{2}\cdots a_{N})}\left|\Sigma^{(k)}\right.\right]\equiv x_{i}^{(a_{1}a_{2}\cdots a_{k})}$
($k=1,2,\ldots,N-1$) are identical for all $a_{k+1},\: a_{k+2},\:\ldots,\: a_{N}$.
Suppose $\mathsf{E}\left[x_{i}^{(a_{1}a_{2}\cdots a_{N})}\right]=m_{i}$,
$\mathsf{V}\left[x_{i}^{(a_{1})}\right]=\sigma_{1}^{2}$, $\mathsf{E}\left[\mathsf{V}\left[x_{i}^{(a_{1}a_{2}\cdots a_{k})}\left|\Sigma^{(k)}\right.\right]\right]=\sigma_{k+1}^{2}$
($k=1,\ldots,N-1$) with real and finite $m_{i}$, $\sigma_{1}^{2}$,
$\sigma_{k+1}^{2}$, . Then the metric on $M_{n}$
\begin{equation}
d_{n}\left(x^{(a_{1}a_{2}\cdots a_{N})},x^{(b_{1}b_{2}\cdots b_{N})}\right)\equiv\dfrac{1}{\sqrt{n}}\sqrt{\sum_{i=1}^{n}\left(x_{i}^{(a_{1}a_{2}\cdots a_{N})}-x_{i}^{(b_{1}b_{2}\cdots b_{N})}\right)^{2}}\label{d_gen}
\end{equation}
has the property
\[
d_{n}\left(x^{(a_{1}a_{2}\cdots a_{N})},x^{(b_{1}b_{2}\cdots b_{N})}\right)\overset{\mathrm{P}}{\longrightarrow}u_{a_{1}a_{2}\cdots a_{N},b_{1}b_{2}\cdots b_{N}}
\]
as $n\to\infty$, where
\[
u_{a_{1}a_{2}\cdots a_{N},b_{1}b_{2}\cdots b_{N}}=\sqrt{2}\left(\left(1-\delta_{a_{1}b_{1}}\delta_{a_{2}b_{2}}\cdots\delta_{a_{N}b_{N}}\right)\sigma_{N}^{2}+\right.
\]

\begin{equation}
\left.+\left(1-\delta_{a_{1}b_{1}}\delta_{a_{2}b_{2}}\cdots\delta_{a_{N-1}b_{N-1}}\right)\sigma_{N-1}^{2}+\cdots+\left(1-\delta_{a_{1}b_{1}}\right)\sigma_{1}^{2}\right)^{\tfrac{1}{2}}.\label{u_a...}
\end{equation}
is nonisometric ultrametric $p_{1}p_{2}\cdots p_{N}\times p_{1}p_{2}\cdots p_{N}$
matrix.

\label{th7} \end{theorem}

The proof \theor{th7} is analogous to the proof of \theor{th6}
and is based on the identity

\[
\mathsf{E}\left[\left(x_{i}^{(a_{1}a_{2}\cdots a_{k})}-x_{i}^{(b_{1}b_{2}\cdots b_{k})}\right)^{2}\right]=
\]

\[
=\left(1-\delta_{a_{1}b_{1}}\delta_{a_{2}b_{2}}\cdots\delta_{a_{k}b_{k}}\right)\mathsf{E}\left[\mathsf{V}\left[x_{i}^{(a_{1}a_{2}\cdots a_{k})}\left|\Sigma^{(k-1)}\right.\right]+\mathsf{V}\left[x_{i}^{(b_{1}b_{2}\cdots b_{k})}\left|\Sigma^{(k-1)}\right.\right]\right]+
\]

\[
+\left(1-\delta_{a_{1}b_{1}}\delta_{a_{2}b_{2}}\cdots\delta_{a_{k-1}b_{k-1}}\right)\mathsf{E}\left[\left(x_{i}^{(a_{1}a_{2}\cdots a_{k-1})}-x_{i}^{(b_{1}b_{2}\cdots b_{k-1})}\right)^{2}\right].
\]
for $k=1,\ldots,N$.

\section{The algorithm of the stochastic generation of ultrametric structures
in high-dimensional spaces and the numerical simulation}

In this section, we describe the procedure for constructing an ultrametric
space which is induced by the special random distribution of points
in the space $R^{n}$ for high $n$. We restrict ourselves to finite
homogeneous ultrametric spaces. Such spaces are equivalent to the
boundary of some $N$-level hierarchical tree with fixed numbers of
branching $p_{k}$ for every $k$-th level. For $p_{1}=p_{2}=\cdots=p_{N}=p$
the homogeneous ultrametric space is self-similar, and the corresponding
hierarchical tree is the $N$-level Cayley tree, with the number of
branching equal to $p$. However, following our procedure, one can
construct a non--homogeneous finite ultrametric space isomorphic to
the boundary of $N$-level hierarchical tree with arbitrary number
of branches at each node. The procedure of building such an ultrametric
space in accordance with \theor{th7} can be described as follows.
We generate $p_{1}$ independent random points $x^{(a_{1})}=\left(x_{1}^{(a_{1})},x_{2}^{(a_{1})},\ldots,x_{n}^{(a_{1})}\right)$
($a_{1}=1,2,\ldots,p_{1}$ ) in the $n$-dimensional space $R^{n}$
with the normal distribution $\mathcal{N}\left(0,\sigma_{1}\right)$
for each coordinate. Next we generate $p_{1}p_{2}$ independent random
points $x^{(a_{1}a_{2})}=\left(x_{1}^{(a_{1}a_{2})},x_{2}^{(a_{1}a_{2})},\ldots,x_{n}^{(a_{1}a_{2})}\right)$
($a_{1}=1,2,\ldots,p_{1}$, $a_{2}=1,2,\ldots,p_{1}p_{2}$) in $R^{n}$
with normal distribution $\mathcal{N}\left(x_{i}^{(a_{1})},\sigma_{2}\right)$
for $i$-th coordinate. Next we generate $p_{1}p_{2}p_{3}$ independent
random points $\left(x_{1}^{(a_{1}a_{2}a_{3})},x_{2}^{(a_{1}a_{2}a_{3})},\right.$$\left.\ldots,x_{n}^{(a_{1}a_{2}a_{3})}\right)$
($a_{1}=1,2,\ldots,p_{1}$, $a_{2}=1,2,\ldots,p_{1}p_{2}$, $a_{3}=1,2,\ldots,p_{1}p_{2}p_{3}$)
in $R^{n}$ with the normal distribution $\mathcal{N}\left(x_{i}^{(a_{1}a_{2})},\sigma_{3}\right)$
for $i$-th coordinate and so on. We repeat this procedure for the
generation of random points $N$ times. On the last step we generate
$p_{1}p_{2}\cdots p_{N}$ independent random points $x^{(a_{1}a_{2}\cdots a_{N})}=$
$\left(x_{1}^{(a_{1}a_{2}\cdots a_{N})},x_{2}^{(a_{1}a_{2}\cdots a_{N})},\right.$
$\left.\ldots,x_{n}^{(a_{1}a_{2}\cdots a_{N})}\right)$ ($a_{1}=1,2,\ldots,p_{1}$,
$a_{2}=1,2,\ldots,p_{1}p_{2}$, $\ldots,$ $a_{N}=1,2,\ldots,p_{1}p_{2}\cdots p_{N}$)
in $R^{n}$ with the normal distribution $\mathcal{N}\left(x_{i}^{(a_{1}a_{2}\cdots a_{N-1})},\sigma_{N}\right)$
for $i$-th coordinate. The set of points $M_{n}^{(N)}=\left\{ x^{(a_{1}a_{2}\cdots a_{N})}\right\} $
forms the metric space with the metric (\ref{d_gen}). The metric
(\ref{d_gen}) will converge in probability for $n\to\infty$ to some
ultrametric in the sense of \theor{th7}.

We introduce some new definitions.

Let $M$ be the finite metric space with elements $x^{(a)}\in R^{n},\: a=1,2,\ldots,N$
and let $d\left(x^{(a)},x^{(b)}\right)$ be the metric on $M$. For
for any triplet of points (or for any triangle) $x^{a},\, x^{b},\, x^{c}$
in $M$ we define two functions:
\[
I(a,b,c)=1-\dfrac{\mathrm{mid}\left\{ d(x^{(a)},x^{(b)}),d(x^{(b)},x^{(c)}),d(x^{(c)},x^{(a)}))\right\} }{\max\left\{ d(x^{(a)},x^{(b)}),d(x^{(b)},x^{(c)}),d(x^{(c)},x^{(a)}))\right\} },
\]
and
\[
J(a,b,c)=1-\dfrac{\min\left\{ d(x^{(a)},x^{(b)}),d(x^{(b)},x^{(c)}),d(x^{(c)},x^{(a)}))\right\} }{\max\left\{ d(x^{(a)},x^{(b)}),d(x^{(b)},x^{(c)}),d(x^{(c)},x^{(a)}))\right\} },
\]
where
\[
\mathrm{mid}\left\{ d(x^{(a)},x^{(b)}),d(x^{(b)},x^{(c)}),d(x^{(c)},x^{(a)}))\right\} \equiv\left(d(x^{(a)},x^{(b)})+d(x^{(b)},x^{(c)})+d(x^{(c)},x^{(a)})-\right.
\]
\[
\left.-\max\left\{ d(x^{(a)},x^{(b)}),d(x^{(b)},x^{(c)}),d(x^{(c)},x^{(a)}))\right\} -\min\left\{ d(x^{(a)},x^{(b)}),d(x^{(b)},x^{(c)}),d(x^{(c)},x^{(a)}))\right\} \right),
\]
and we call this the middle length of the of the triangle $\left\{ x^{(a)},x^{(b)},x^{(c)}\right\} $.
We say that $I(a,b,c)$ is the \emph{ultrametricity index} and $I(a,b,c)$
is the \emph{isometricity index} of the triangle $\left\{ x^{(a)},x^{(b)},x^{(c)}\right\} $.

\emph{The ultrametricity index of the metric space $M$} is the number
$U\left(M\right)$ defined as the average ultrametricity index of
all possible triangles in $M$:
\[
U\left(M\right)=\frac{3!(N-3)!}{N!}\sum_{a=1}^{N}\sum_{b=a+1}^{N}{\rm \sum_{c=b+1}^{N}}I(a,b,c).
\]

\emph{The isometricity index of the metric space $M$} is the number
$E\left(M\right)$ defined as the average isometricity index of all
possible triangles in $M$:
\[
E\left(M\right)=\frac{3!(N-3)!}{N!}\sum_{a=1}^{N}\sum_{b=a+1}^{N}{\rm \sum_{c=b+1}^{N}}I(a,b,c).
\]

Note that $0\leq U\left(M\right)\leq\dfrac{1}{2}$, $0\leq E\left(M\right)\leq1$
and we have $U\left(M\right)=0$ if $M$ is ultrametric and $E\left(M\right)=0$
if $M$ is isometric. It is possible to claim that for metric space
$M_{n}^{(N)}$ generated by the procedure described in the previous
section
\[
U\left(M_{n}^{(N)}\right)\overset{\mathrm{P}}{\rightarrow}0
\]
as $n\to\infty$.

Let $d_{n}^{\left(p_{1},p_{2}, \dots, p_{N}\right)}\left(N\right)$ be the metric matrix generated
numerically according to the above procedure, where $N$ is the number
of procedure steps (the level number of the ultrametric tree), $p_{i}$,
$i=1,2, \dots , N$ is the numbers of points at the $i$-th step of the procedure. Below we present the metric matrix $d_{n}^{\left(p_{1},p_{2}\right)}\left(2\right)$
and $d_{n}^{\left(p_{1},p_{2},p_{3}\right)}\left(3\right)$.
The values of the variances are $\sigma_{1}^{2}=\sigma_{2}^{2}=\sigma_{3}^{2}=1$.
The dimension $n$ of the space $R^{n}$ is chosen to be $10^{4}$.

\[
d^{\left(3,3\right)}\left(2\right)=\left(\begin{array}{ccccccccc}
0 & 1.4055 & 1.4199 & 2.0049 & 1.9981 & 1.9989 & 1.9968 & 1.9958 & 2.0010\\
1.4055 & 0 & 1.4116 & 2.0141 & 1.9973 & 1.9984 & 2.0133 & 2.0173 & 2.0211\\
1.4199 & 1.4116 & 0 & 2.0148 & 2.0082 & 2.0022 & 2.0240 & 2.0245 & 2.0275\\
2.0049 & 2.0141 & 2.0148 & 0 & 1.4166 & 1.4163 & 1.9921 & 1.9770 & 1.9958\\
1.9981 & 1.9973 & 2.0082 & 1.4166 & 0 & 1.4137 & 1.9955 & 1.9794 & 1.9921\\
1.9989 & 1.9984 & 2.0022 & 1.4163 & 1.4137 & 0 & 1.9894 & 1.9592 & 1.9809\\
1.9968 & 2.0133 & 2.0240 & 1.9921 & 1.9955 & 1.9894 & 0 & 1.4124 & 1.4019\\
1.9958 & 2.0173 & 2.0245 & 1.9770 & 1.9794 & 1.9592 & 1.4124 & 0 & 1.4027\\
2.0010 & 2.0211 & 2.0275 & 1.9958 & 1.9921 & 1.9809 & 1.4019 & 1.4027 & 0
\end{array}\right)
\]

\[
d^{\left(2,2,2\right)}\left(3\right)=\left(\begin{array}{cccccccc}
0 & 1.4168 & 2.0359 & 2.0355 & 2.4578 & 2.4540 & 2.4434 & 2.4260\\
1.4168 & 0 & 2.0113 & 2.0119 & 2.4398 & 2.4402 & 2.4433 & 2.4147\\
2.0359 & 2.0113 & 0 & 1.4238 & 2.4573 & 2.4580 & 2.4466 & 2.4422\\
2.0355 & 2.0119 & 1.4238 & 0 & 2.4572 & 2.4541 & 2.4472 & 2.4467\\
2.4578 & 2.4398 & 2.4573 & 2.4572 & 0 & 1.4294 & 2.0108 & 2.4572\\
2.4540 & 2.4402 & 2.4580 & 2.4541 & 1.4294 & 0 & 2.0146 & 2.0033\\
2.4434 & 2.4433 & 2.4466 & 2.4472 & 2.0108 & 2.0146 & 0 & 1.4142\\
2.4260 & 2.4147 & 2.4422 & 2.4467 & 2.0005 & 2.0033 & 1.4142 & 0
\end{array}\right)
\]

In full agreement with \theor{th6} and \theor{th7} the matrix
$d^{\left(3,3\right)}\left(2\right)$ and $d^{\left(2,2,2\right)}\left(3\right)$
coincides with the matrix (\ref{u_aabb}) for $\sigma_{1}=\sigma_{2}=1$
\[
u_{a_{1}a_{2},b_{1}b_{2}}=\sqrt{2}\left(\left(1-\delta_{a_{1}b_{1}}\delta_{a_{2}b_{2}}\right)+\left(1-\delta_{a_{1}b_{1}}\right)\right)^{\tfrac{1}{2}}
\]
and (\ref{u_a...}) for $N=3$ and $\sigma_{1}=\sigma_{2}=\sigma_{3}=1$
\[
u_{a_{1}a_{2}a_{3},b_{1}b_{2}b_{3}}=\sqrt{2}\left(\left(1-\delta_{a_{1}b_{1}}\delta_{a_{2}b_{2}}\delta_{a_{3}b_{3}}\right)+\left(1-\delta_{a_{1}b_{1}}\delta_{a_{2}b_{2}}\right)+\left(1-\delta_{a_{1}b_{1}}\right)\right)^{\tfrac{1}{2}}
\]
respectively up to the third digit in the values of its elements.

Below we present a plots of the ultrametricity index $U\left(M\right)$
and isometricity index $E\left(M\right)$ from the space dimension
$n$ for the three step procedure of stochastic generation of ultrametric
space (\fig{figure}). The values of parameters are $p_{1},=p_{2}=p_{3}=3$,
$\sigma_{1}=\sigma_{2}=\sigma_{3}=1$. As one can see from \fig{figure}
even at dimensions $n\gtrsim15$ the metric $27\times27$ matrix can
be satisfactorily considered as an ultrametric matrix.


\begin{figure}[H]
\begin{minipage}[c]{1\linewidth}%
\center{\includegraphics[width=1\linewidth]{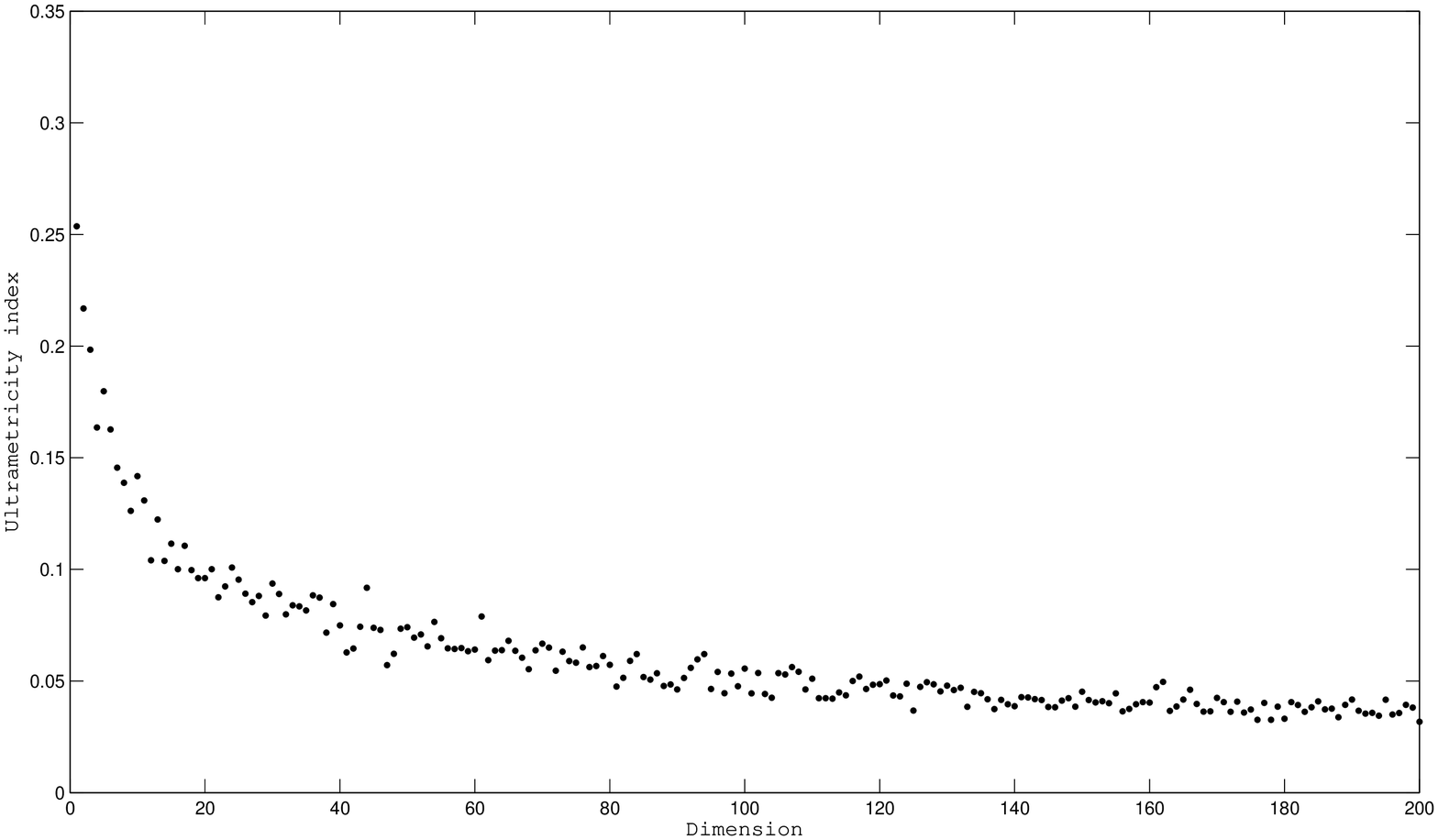}}\\
 (a) \label{figure1} %
\end{minipage}\vfill{}
\begin{minipage}[c]{1\linewidth}%
\center{\includegraphics[width=1\linewidth]{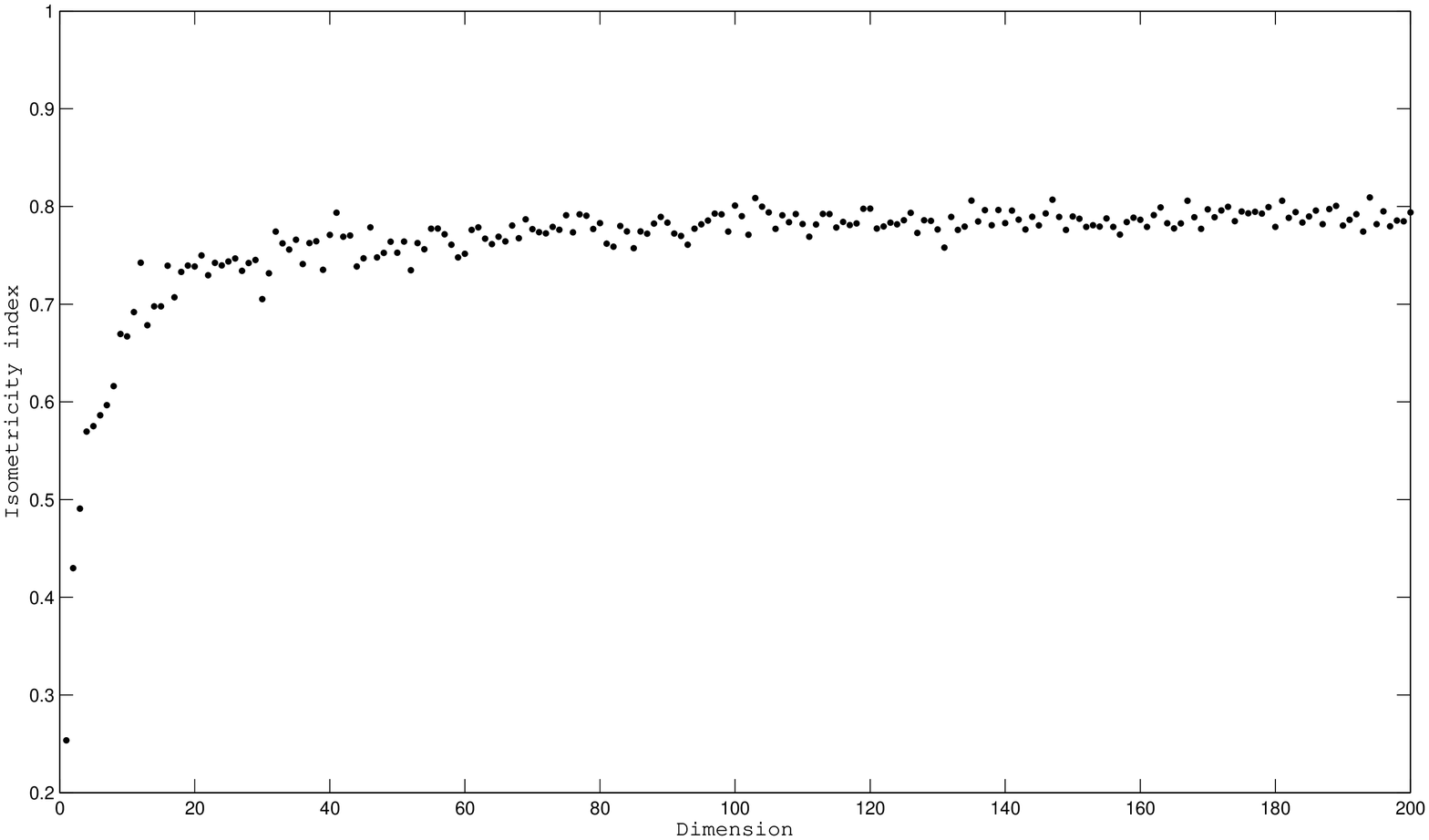}}\\
 (b) %
\end{minipage}\caption{ plot of the ultrametricity index $U\left(M\right)$ (a) and plot
of the isometricity index $E\left(M\right)$ (b) from the space dimension
$n$ for three step procedure of stochastic generation of ultrametric
space in $R^{n}$.}

\label{figure}
\end{figure}


\section{Conclusion}

In this paper a probabilistic mechanism for generating an ultrametrics
in Euclidean metric spaces of high dimension is described. It is based
on a special probability distribution of statistically independent
random points. It is proved that the Euclidean metric on a set of
independent random points in $R^{n}$ with a special distribution
convergence in probability to ultrametric as $n\to\infty$. Ultrametric
distance matrix is completely determined by variances of distributions
of coordinates of points. The probabilistic algorithm for the generation
of finite ultrametric structures of any topology in high-dimensional
Euclidean spaces is described. The validity of the algorithm is demonstrated
by the explicit calculations of distance matrices and ultrametricity
indexes for different dimensions.

The main result of this paper is \theor{th7}. Note that this theorem
like the previous ones since \theor{th4} can be generalized. For
example, there is no need to require that the expectations of conditional
variances $\mathsf{E}\left[\mathsf{V}\left[x_{i}^{(a_{1}a_{2}\cdots a_{N})}\left|\Sigma^{(k-1)}\right.\right]\right]=\sigma_{k,i}^{2}$
of random coordinates of points $x^{(a_{1}a_{2}\cdots a_{k})}$ are
independent on $i$. It is sufficient to require that the limit $\lim_{n\to\infty}\frac{1}{n}\sum_{i=1}^{n}(\sigma_{k,i})^{2}$
is finite. Furthermore one can consider the case when the coordinates
of random points are correlated in some way. We left the generalization
of these theorems for future.

We emphasize that the problem discussed in the present paper is not
about a mechanism of stochastic generation of hidden ultrametric structures
being realized in real high-dimensional complex systems, such as disordered
ferromagnets, spin glasses, etc. We have described only a possible
scenario of appearance of ultrametric for such systems. But even if
there is an implementation of a similar scenario, it is difficult
to track it because of limitation of analytical methods of modeling
and research of these systems. However it seems reasonable that this
mechanism could be realized in complex disordered systems consisting
of a large number of elements interacting in a random way. Therefore
the study of toy models of disordered systems with a special type
of random interactions allowing an accurate analytical study for the
stochastic generation of hidden ultrametric structures in the line
with the described mechanism is of particular interest. 
\vskip 0.5cm \textbf{\large{{{Acknowledgements}}}}{\large{{
\vskip 0.2cm The author is grateful to Vladik Avetisov for useful
discussions. Work was partially supported by the RFBR grant (project
number 13-01-00790-a). }}}{\large \par}


\begin{thebibliography}{10}
\bibitem{S}Schikhof W.H. Ultrametric Calculus. An Introduction to
$p$-adic Analysis, Cambridge Studies in Advanced Mathematics, Cambridge
University Press, Cambridge, 1984. viii+306 pp.

\bibitem{VVZ} Vladimirov V. S., Volovich I. V., Zelenov E. I., $p$-Adic
Analysis and Mathematical Physics, Singapure: World Scientific Publishing,
1994, 340 p.

\bibitem{Sh_Kh} Albeverio S., Khrennikov A.Yu., Shelkovich V.M.,
Theory of $p$-adic Distributions: Linear and Nonlinear Models, London
Mathematical Society Lecture Note Series, 2011, 351 p.

\bibitem{DZ} Dolgopolov M. V., Zubarev A. P., \textquotedbl{}Some
Aspects of $m$-Adic Analysis and Its Applications to $m$-Adic Stochastic
Processes\textquotedbl{}, \emph{$p$-Adic Numbers, Ultrametric Analysis,
and Applications}, 3 (2011), N.1, 39--51.

\bibitem{RTV} Rammal R., Toulose G., Virasoro M. A., \textquotedbl{}Ultrametrisity
for physicists\textquotedbl{}, \emph{Rev. Mod. Phys.}, V. 58 (1986),
N3., 765--788.

\bibitem{ALL} Dragovich B., Khrennikov A.Yu., Kozyrev S.V., Volovich
I.V., \textquotedbl{}On $p$-adic mathematical physics\textquotedbl{}.
\emph{$p$-Adic Numbers, Ultrametric Analysis, and Applications},
Vol.~1, No 1, (2009) 1.

\bibitem{MPSTV} Mezard M., Parisi G., Sourlas N., Toulouse G., Virasoro
M.A., \textquotedbl{}On the nature of the spin-glass phase\textquotedbl{},
\emph{Phys. Rev. Lett.} 52 (1984) 1156.

\bibitem{Parisi1} Parisi G., \textquotedbl{}Infinite number of order
parameters for spin-glasses\textquotedbl{}, \emph{Phys. Rev. Lett.},
43 (1979), 1754--1756.

\bibitem{Dayson} Dayson F. J., \textquotedbl{}Existence of a phase--transition
in a one--dimentional Ising ferromagnet\textquotedbl{}, \emph{Commun.
Math. Phys.}, 12 (1969), 91--107.

\bibitem{Dot}Dotsenko V, An Introduction to the Spin Glasses and
Neural Networks, Singapure: World Scientific, 1994, 159 p.

\bibitem{Frauen1} Frauenfelder H., \textquotedbl{}The connection
between low-temperature kinetics and life\textquotedbl{}, \emph{In
Protein Structure, Molecular and Electronic Reactivity}, Eds. R. H.
Austin et al., Springer, NY, (1987), 245--261.

\bibitem{Frauen2} Frauenfelder H., \textquotedbl{}Complexity in proteins\textquotedbl{},
\emph{Nature Struct. Biol.}, V. 2 (1995), 821--823.

\bibitem{OS} Ogielski A. T., Stein D. L., \textquotedbl{}Dynamics
on ultrametric spaces\textquotedbl{}, \emph{Phys. Rev. Lett.}, V.
55 (1985), 1634--1637.

\bibitem{BK} Becker O. K., Karplus M., \textquotedbl{}The topology
of multidimentional protein enargy surfaces: theory and application
to peptide structure and kinetics\textquotedbl{}, \emph{J. Chem. Phys.},
106 (1997), 1495--1517.

\bibitem{ABK} Avetisov V. A., Bikulov A. Kh., Kozyrev, S. V., \textquotedbl{}Application
of $p$--adic analysis to models of spontaneous breaking of replica
symmetry\textquotedbl{}, \emph{Journal of Physics A}, Vol.32 (1999),
8785--8791.

\bibitem{ABKO} Avetisov V. A., Bikulov A. H., Kozyrev S. V, Osipov
V. A., \textquotedbl{}$p$-Adic models of ultrametric diffusionh constrained
by hierarchical energy landscapes\textquotedbl{}, \emph{J. Phys. A:
Math. Gen.}, V. 35 (2002), 177--189.

\bibitem{ABO} Avetisov V. A., Bikulov A. Kh., Osipov V. A., \textquotedbl{}$p$-Adic
description of characteristic relaxation in complex systems\textquotedbl{},
\emph{J. Phys. A: Math. Gen.}, 36 (2003), 4239--4246.

\bibitem{AB} Avetisov, V. A., Bikulov, A. Kh., \textquotedbl{}Protein
ultrametricity and spectral diffusion in deeply frozen proteins\textquotedbl{},
\emph{Biophysical Reviews and Letters}, 3 (2008) 387--396.

\bibitem{ABZ} Avetisov V. A, Bikulov A. H., Zubarev A. P., \textquotedbl{}First
passage time distribution and number of returns for ultrametric random
walk\textquotedbl{}, \emph{J. Phys. A: Math and Theor.}, V. 42 (2009),
85005--85021.

\bibitem{SJ}Sornette D., Johansen A., \textquotedbl{}A hierarchical
model of financial crashes\textquotedbl{}, \emph{Physica A}, V. 261,
1998, P. 581--598.

\bibitem{MS} Mantenga R.N., Stanley H.E, An Introduction to Econophysics.
Correlations and Complexity in Finance, Cambridge Univ. Press, Cambridge,
2000, 147 p.

\bibitem{BZK} Bikulov A.Kh., Zubarev A.P., Kaidalova L.V., \textquotedbl{}Hierarchic
Dynamic Model of Financial Market near Crash Points and $p$-adic
Mathematical Analysis\textquotedbl{}, \emph{Proc. Samara State Tech.
Univ.}, Iss. 42, (2006) 135-141 (in Russian).

\bibitem{V} Volov V.T., \textquotedbl{}Fractal-cluster theory and
thermodynamic principles of the control and analysis for the self-organizing
systems\textquotedbl{}, arXiv:1309.1415.

\bibitem{VZ} Volov V.T., Zubarev A.P., \textquotedbl{}Ultrametric
dynamics for the closed fractal-cluster resource models\textquotedbl{},
arXiv:1309.4116.

\bibitem{Hall} P. Hall, J.S. Marron and A. Neeman, \textquotedbl{}Geometric
representation of high dimension low sample size data\textquotedbl{},
\emph{Journal of the Royal Statistical Society B}, 67, 427--444, 2005.

\bibitem{Murtagh} F. Murtagh, \textquotedbl{}From data to the physics
using ultrametrics: new results in high dimensional data analysis\textquotedbl{},
in B. Dragovich, A. Khrennikov, Z. Rakic and I. Volovich, Eds., \emph{Proc.
2nd International Conference on p-Adic Mathematical Physics}, American
Institute of Physics, 151-161, 2006.

\bibitem{Dur} Durrett R., Probability. Theory and examples, Wadsworth
\& Brooks/Cole, Pacific Grove, CA, 1991, 453 pp.

\bibitem{Bill1} Billingslety P., Probability and measure, New York:
Wiley, 1979, xiv+515 pp.

\bibitem{Shiryaev} Shiryaev A.N., Probability, Graduate Texts in
Mathematics, 95, Second edition, Springer-Verlag, New York, 1996,
xvi+623 pp.

\bibitem{Sl} Slutsky E.E., \textquotedbl{}Uber Stochastische Asymptoten
und Grenzwerte\textquotedbl{}, Metron 5, 1--90.

\bibitem{Bill} Billingsley P., Convergence of Probability Measures,
New York, NY: John Wiley \& Sons, Inc., 1999, xii+253 pp.\end{thebibliography}
\end{document}